\pdfoutput=1
\documentclass[sigconf,authorversion]{acmart}

\usepackage[english]{babel}  
\usepackage{amssymb}
\usepackage{boxedminipage} 

\usepackage[shortlabels]{enumitem}

\usepackage[capitalise]{cleveref}

\makeatletter
\def\renewtheorem#1{%
  \expandafter\let\csname#1\endcsname\relax
  \expandafter\let\csname c@#1\endcsname\relax
  \gdef\renewtheorem@envname{#1}
  \renewtheorem@secpar
}
\def\renewtheorem@secpar{\@ifnextchar[{\renewtheorem@numberedlike}{\renewtheorem@nonumberedlike}}
\def\renewtheorem@numberedlike[#1]#2{\newtheorem{\renewtheorem@envname}[#1]{#2}}
\def\renewtheorem@nonumberedlike#1{  
\def\renewtheorem@caption{#1}
\edef\renewtheorem@nowithin{\noexpand\newtheorem{\renewtheorem@envname}{\renewtheorem@caption}}
\renewtheorem@thirdpar
}
\def\renewtheorem@thirdpar{\@ifnextchar[{\renewtheorem@within}{\renewtheorem@nowithin}}
\def\renewtheorem@within[#1]{\renewtheorem@nowithin[#1]}

\newtheoremstyle{framedthmenv}%
  {0cm}
  {0cm}
  {\@acmdefinitionbodyfont}
  {\@acmdefinitionindent}
  {\@acmdefinitionheadfont}
  {:}
  {.5em}
  {\thmname{#1}\thmnumber{ #2}\thmnote{ {\@acmdefinitionnotefont(#3)} }}
\makeatother

\theoremstyle{acmplain}
\renewtheorem{theorem}{Theorem}[section]
\renewtheorem{proposition}[theorem]{Proposition}
\renewtheorem{lemma}[theorem]{Lemma}
\renewtheorem{corollary}[theorem]{Corollary}
\theoremstyle{acmdefinition}
\renewtheorem{definition}[theorem]{Definition}
\renewtheorem{remark}[theorem]{Remark}

\theoremstyle{framedthmenv}
\newtheorem{problem}{Problem} 
\crefname{problem}{Problem}{Problems}
\Crefname{problem}{Problem}{Problems}
\newtheorem{algorithm}{Algorithm} 
\crefname{algorithm}{Algorithm}{Algorithms}
\Crefname{algorithm}{Algorithm}{Algorithms}

\theoremstyle{acmplain}
\newcommand{\bigO}[1]{O(#1)} 
\newcommand{\bigOPar}[1]{O\!\left(#1\right)} 
\newcommand{\softO}[1]{O\tilde{~}(#1)} 
\newcommand{\expmatmul}{\omega} 
\newcommand{\costLQUP}[2]{\mathcal{C}(#1,#2)} 
\newcommand{\ZZ}{\mathbb{Z}} 
\newcommand{\ZZp}{\mathbb{Z}_{> 0}} 
\newcommand{\tuple}[1]{\mathbf{#1}}  
\newcommand{\var}{X} 
\newcommand{\field}{\mathbb{K}} 
\newcommand{\polRing}{\field[\var]} 
\newcommand{\rdim}{m} 
\newcommand{\cdim}{n} 
\newcommand{\storeArg}{} 
\newcommand{\matSpace}[1][\rdim]{\renewcommand\storeArg{#1}\matSpaceAux} 
\newcommand{\polMatSpace}[1][\rdim]{\renewcommand\storeArg{#1}\polMatSpaceAux} 
\newcommand{\matSpaceAux}[1][\storeArg]{\field^{\storeArg \times #1}} 
\newcommand{\polMatSpaceAux}[1][\storeArg]{\polRing^{\storeArg \times #1}} 
\newcommand{\row}[1]{ \mathbf{\MakeLowercase{#1}} } 
\newcommand{\col}[1]{ \mathbf{\MakeLowercase{#1}} } 
\newcommand{\mat}[1]{ \mathbf{\MakeUppercase{#1}} } 
\newcommand{\matt}[1]{\mathbf{ \hat{\MakeUppercase{#1}} }} 
\newcommand{\matz}{\mat{0}} 
\newcommand{\sumVec}[1]{|#1|} 
\newcommand{\trsp}[1]{#1^\mathsf{T}} 
\newcommand{\matrow}[2]{{#1}_{#2,*}} 
\newcommand{\matcol}[2]{{#1}_{*,#2}} 
\newcommand{\submat}[3]{{#1}_{#2,#3}} 
\newcommand{\diag}[1]{\mathrm{diag}(#1)}  
\newcommand{\idMat}[1][\rdim]{\mat{I}_{#1}} 
\newcommand{\rank}[1]{\operatorname{rank}(#1)}  
\newcommand{\rdeg}[2][]{\mathrm{rdeg}_{{#1}}(#2)} 
\newcommand{\cdeg}[2][]{\mathrm{cdeg}_{{#1}}(#2)} 
\newcommand{\xDiag}[1]{\mat{\var}^{#1\,}} 
\newcommand{\shiftSpace}[1][\rdim]{\ZZ^{#1}} 
\newcommand{\shift}[2][s]{#1_{#2}} 
\newcommand{\shifts}[1][s]{\mathbf{#1}} 
\newcommand{\size}[1]{\mathrm{Size}(#1)} 
\newcommand{\vsdim}{D} 
\newcommand{\order}{d} 
\newcommand{\orders}{\tuple{\order}} 
\newcommand{\orderSpace}[1][\nbeq]{\ZZp^{#1}} 
\newcommand{\nbeq}{\cdim} 
\newcommand{\nbun}{\rdim} 
\newcommand{\sys}{\mat{F}} 
\newcommand{\sysSpace}{\polMatSpace[\nbun][\nbeq]} 
\newcommand{\app}{\row{p}} 
\newcommand{\appSpace}{\polMatSpace[1][\nbun]} 
\newcommand{\appbas}{\mat{P}} 
\newcommand{\appbasSpace}{\polMatSpace[\nbun]} 
\newcommand{\mods}[1][\orders]{\mat{\var}^{#1\,}} 
\newcommand{\modAppCustom}[2]{{\mathcal A}_{#1}(#2)}
\newcommand{\modApp}{{\mathcal A}_{\orders}(\sys)}
 
\newcommand{\cmat}{\mat{C}} 
\newcommand{\cmatSpace}{\matSpace[\nbun][\nbeq]} 
\newcommand{\syss}{\mat{H}} 
\newcommand{\truncOrd}{t} 
\newcommand{\truncOrds}{\tuple{\truncOrd}} 
\newcommand{\truncOrdSpace}[1][\nbeq]{\ZZp^{#1}} 
\newcommand{\fsubset}{S} 
\newcommand{\rhs}{\mat{G}} 
\newcommand{\rpol}{f} 
\newcommand{\evpt}{\alpha} 
\newcommand{\rem}{\mathbin{\mathrm{rem}}} 
\newcommand{\dmax}{\delta} 
\newcommand{\uvec}{\row{u}} 
\newcommand{\uvecSpace}{\matSpace[1][\nbun]} 
\newcommand{\sumOrds}{\sumVec{\truncOrds}} 
\newcommand{\hvec}[1]{\row{c}_{#1}} 
\newcommand{\largeRows}[1]{\mathcal{R}_{#1}} 
\newcommand{\largeCols}[1]{\mathcal{C}_{#1}} 
\newcommand{\largeOrders}[1]{\mathcal{D}_{#1}} 
\newcommand{\algoname}[1]{{\normalfont\textsc{#1}}}
\newcommand{\problemname}[1]{{\normalfont\textsc{#1}}}
\newcommand{\algoword}[1]{\emph{\textsf{#1}}}
\newcommand{\assign}{\leftarrow}
\newcommand{\inlcomment}[1]{\texttt{\small/* #1 */}}
\newcommand{\eolcomment}[1]{\hfill\texttt{\small// #1}}
\newenvironment{algobox}
{
  \newcommand{\algoInfo}[2]{
    \begin{algorithm}
    \label{##2}
    \emph{\algoname{##1}}
  }
  \newcommand{\dataInfos}[2]{
    \algoword{##1:}
      \begin{itemize}[leftmargin=0.8cm]
          ##2
      \end{itemize}
  }
  \newcommand{\dataInfo}[2]{
    \algoword{##1:} ##2 }
  \newcommand{\algoSteps}[1]{
    \setlist[enumerate,1]{leftmargin=0.5cm}
    \setlist[enumerate,2]{leftmargin=0.4cm}
    \setlist[enumerate,3]{leftmargin=0.4cm}
    \begin{enumerate}[label=\textbf{\arabic*.}]
        ##1
    \end{enumerate}
  }
  \begin{figure}[ht]
  \centering
  \addtolength\fboxsep{0.1cm}
  \begin{boxedminipage}{0.99\columnwidth}
}
{
  \end{algorithm}
  \end{boxedminipage}
  \end{figure}
}
\newenvironment{problembox}{
  \newcommand{\problemInfo}[2]{
    \begin{problem}
    \label{##2}
    \problemname{##1}
  }
  \newcommand{\dataInfos}[2]{
    \emph{##1:}
      \begin{itemize}[leftmargin=0.8cm]
          ##2
      \end{itemize}
  }
  \newcommand{\dataInfo}[2]{
   \emph{##1:} 
      \begin{itemize}[leftmargin=0.8cm]
          \item ##2
        \end{itemize}
  }

  \begin{figure}[!ht]
  \centering
  \addtolength\fboxsep{0.1cm}
  \begin{boxedminipage}{0.99\columnwidth}
  }
  {
  \end{problem}
  \end{boxedminipage}
  \end{figure}
}

\newcommand{\myparagraph}[1]{ \paragraph{\hspace{-0.35cm}\textbf{#1}} }  

\begin{document}

\fancyhead{} 
\title{Certification of Minimal Approximant Bases}

\author{Pascal Giorgi}
\affiliation{%
  \institution{LIRMM, Universit\'e de Montpellier, CNRS}
  \city{Montpellier}
  \state{France}
}
\email{pascal.giorgi@lirmm.fr}

\author{Vincent Neiger}
\affiliation{%
  \institution{Univ. Limoges, CNRS, XLIM, UMR\,7252}
  \city{F-87000 Limoges} 
  \state{France} 
}
\email{vincent.neiger@unilim.fr}

\copyrightyear{2018}
\acmYear{2018}
\setcopyright{acmlicensed}
\acmConference[ISSAC '18]{2018 ACM International Symposium on Symbolic and Algebraic Computation}{July 16--19, 2018}{New York, NY, USA \\ \phantom{bla} \\ \phantom{bla} \\ \phantom{bla}}
\acmBooktitle{ISSAC '18: 2018 ACM International Symposium on Symbolic and Algebraic Computation, July 16--19, 2018, New York, NY, USA}
\acmPrice{15.00}
\acmDOI{10.1145/3208976.3208991}
\acmISBN{978-1-4503-5550-6/18/07}

\begin{abstract}
  For a given computational problem, a certificate is a piece of data that one
  (the \emph{prover}) attaches to the output with the aim of allowing efficient
  verification (by the \emph{verifier}) that this output is correct.  Here, we
  consider the minimal approximant basis problem, for which the fastest known
  algorithms output a polynomial matrix of dimensions $\nbun\times\nbun$ and
  average degree $\vsdim/\nbun$ using $\softO{\nbun^\expmatmul
  \frac{\vsdim}{\nbun}}$ field operations.  We propose a certificate which, for
  typical instances of the problem, is computed by the prover using
  $\bigO{\nbun^\expmatmul \frac{\vsdim}{\nbun}}$ additional field operations
  and allows verification of the approximant basis by a Monte Carlo algorithm
  with cost bound $\bigO{\nbun^\expmatmul + \nbun\vsdim}$.

  Besides theoretical interest, our motivation also comes from the fact that
  approximant bases arise in most of the fastest known algorithms for linear
  algebra over the univariate polynomials; thus, this work may help in
  designing certificates for other polynomial matrix computations.
  Furthermore, cryptographic challenges such as breaking records for discrete
  logarithm computations or for integer factorization rely in particular on
  computing minimal approximant bases for large instances: certificates can
  then be used to provide reliable computation on outsourced and error-prone
  clusters.
\end{abstract}

%

\keywords{Certification; minimal approximant basis; order basis; polynomial
matrix; truncated product.}

\maketitle

\section{Introduction}
\label{sec:intro}

\myparagraph{Context.}

For a given tuple $\orders= (\order_1,\ldots,\order_\nbeq) \in \orderSpace$
called \emph{order}, we consider an $\nbun\times\nbeq$ matrix $\sys$ of formal
power series with the column $j$ truncated at order $\order_j$.  Formally, we
let $\sys \in \sysSpace$ be a matrix over the univariate polynomials over a
field $\field$, such that the column $j$ of $\sys$ has degree less than
$\order_j$.  Then, we consider the classical notion of minimal approximant
bases for $\sys$ \cite{BarBul92,BecLab94}.  An approximant is a polynomial row
vector $\app \in \appSpace$ such that
\begin{equation}
  \app \sys = \matz \bmod \mods,
  \quad \text{where}\;\;
  \mods = \diag{\var^{\order_1},\ldots,\var^{\order_\nbeq}};
\end{equation} 
here $\app \sys = \matz \bmod \mods$ means that $\app \sys = \col{q} \mods$ for
some $\col{q} \in \polMatSpace[1][\nbeq]$.  The set of all approximants forms a
(free) $\polRing$-module of rank $\nbun$,
\[
  \modApp = \left\{ \app \in \appSpace \,\bigm|\, \app \sys = \matz \bmod \mods \right\}.
\]
A basis of this module is called an \emph{approximant basis} (or sometimes an
\emph{order basis} or a \emph{$\sigma$-basis}); it is a nonsingular matrix in
$\appbasSpace$ whose rows are approximants in $\modApp$ and generate $\modApp$.

The design of fast algorithms for computing approximant bases has been studied
throughout the last three decades
\cite{BarBul92,BecLab94,GiJeVi03,Storjohann06,ZhoLab12,JeNeScVi16}.
Furthermore, these algorithms compute \emph{minimal} bases, with respect to
some degree measure specified by a shift $\shifts\in\shiftSpace$.  The best
known cost bound is $\softO{\nbun^{\expmatmul-1} \vsdim}$ operations in
$\field$ \cite{JeNeScVi16} where $\vsdim$ is the sum $\vsdim=\sumVec{\orders} =
\order_1+\cdots+\order_\nbeq$.
Throughout the paper, our complexity estimates will fit the algebraic RAM model
counting only operations in $\field$, and we will use $\bigO{n^\omega}$ to
refer to the complexity of the multiplication of two $\nbun\times\nbun$
matrices, where $\omega < 2.373$ \cite{CopWin90,LeGall14}.

Here, we are interested in the following question:
\begin{quote}
  \emph{How to efficiently certify that some approximant basis algorithm indeed
  returns an $\shifts$-minimal basis of $\modApp$?}
\end{quote}
Since all known fast approximant basis algorithms are deterministic, it might
seem that a posteriori certification is pointless.  In fact, it is an essential
tool in the context of unreliable computations that arise when one delegates
the processing to outsourced servers or to some large infrastructure that may
be error-prone.  In such a situation, and maybe before concluding a commercial
contract to which this computing power is attached, one wants to ensure that he
will be able to guarantee the correctness of the result of these computations.
Of course, to be worthwhile, the verification procedure must be significantly
faster than the original computation.

Resorting to such computing power is indeed necessary in the case of large
instances of approximant bases, which are a key tool within challenging
computations that try to tackle the hardness of some cryptographic protocols,
for instance those based on the discrete logarithm problem (e.g. El Gamal) or
integer factorization (e.g. RSA).  The computation of a discrete logarithm over
a $768$-bit prime field, presented in \cite{DLP2017}, required to compute an
approximant basis that served as input for a larger computation which took a
total time of 355 core years on a 4096-cores cluster.  The approximant basis
computation itself took 1 core year.  In this context, it is of great interest
to be able to guarantee the correctness of the approximant basis before
launching the most time-consuming step. 

Linear algebra operations are good candidates for designing fast verification
algorithms since they often have a cost related to matrix multiplication while
their input only uses quadratic space.  The first example one may think of is
linear system solving.  Indeed, given a solution vector $\col{x}\in\field^n$ to
a system $\mat{A} \col{x} = \col{b}$ defined by $\mat{A}\in\field^{n\times n}$
and $\col{b} \in\field^{n}$, one can directly verify the correctness by
checking the equations at a cost of $\bigO{n^2}$ operations in $\field$.
Comparatively, solving the system with the fastest known algorithm costs
$\bigO{n^\omega}$.

Another famous result, due to Freivalds \cite{Freivalds1979}, gives a method to
verify a matrix product. Given matrices $\mat{A},\mat{B},\mat{C}\in
\matSpace[n]$, the idea is to check $\row{u}\mat{C}= (\row{u}\mat{A})\mat{B}$
for a random row vector $\row{u}\in\{0,1\}^{1\times n}$, rather than
$\mat{C}=\mat{A}\mat{B}$.  This verification algorithm costs $\bigO{n^2}$ and
is \emph{false-biased one-sided} Monte-Carlo (it is always correct when it
answers ``false''); the probability of error can be made arbitrarily small by
picking several random vectors.

In some cases, one may require an additional piece of data to be produced
together with the output in order to prove the correctness of the result.  For
example, Farkas' lemma \cite{Farkas1902} certifies the infeasibility of a
linear program thanks to an extra vector.  Although the verification is
deterministic in this example, the design of certificates that are verified by
probabilistic algorithms opened a line of work for faster certification methods in
linear algebra \cite{KalNehSau:2011,Kaltofen2012,DumKal:2014,DuKaThVi:2016}. 

In this context, one of the main challenges is to design \emph{optimal}
certificates, that is, ones which are verifiable in linear time. Furthermore,
the time and space needed for the certificate must remain negligible.  In this
work, we seek such an optimal certificate for the problem of computing shifted
minimal approximant bases.

Here, an instance is given by the input $(\orders,\sys,\shifts)$ which is of
size $\bigO{\nbun\vsdim}$: each column $j$ of $\sys$ contains at most
$\nbun\order_j$ elements of $\field$, and the order sums to
$\order_1+\cdots+\order_\nbeq = |\orders|=\vsdim$.  We neglect the size of the
shift $\shifts$, since one may always assume that it is nonnegative and such
that $\max(\shifts)<\nbun\vsdim$ (see \cite[App.\,A]{JeNeScVi16}).  Thus,
ideally one would like to have a certificate which can be verified in time
$\bigO{\nbun\vsdim}$.

In this paper, we provide a non-interactive certification protocol which uses
the input $(\orders,\sys,\shifts)$, the output $\appbas$, and a certificate
which is a constant matrix $\cmat \in \cmatSpace$.  We design a Monte-Carlo
verification algorithm with cost bound $\bigO{\nbun\vsdim + \nbun^{\omega-1}
(\nbun+\nbeq)}$; this is optimal as soon as $\vsdim$ is large compared to
$\nbun$ and $\nbeq$ (e.g. when $\vsdim > \nbun^2 + \nbun\nbeq$), which is most
often the case of interest.  We also show that the certificate $\cmat$ can be
computed in $\bigO{\nbun^{\expmatmul-1}\vsdim}$ operations in $\field$, which
is faster than known approximant basis algorithms.

\myparagraph{Degrees and size of approximant bases.}

For $\appbas\in\appbasSpace$, we denote the row degree of $\appbas$ as
$\rdeg{\appbas} =(r_1,\ldots,r_\nbun)$ where $r_i=\deg(\appbas_{i,*})$ is the
degree of the row $i$ of $\appbas$ for $1\leq i\leq \nbun$.  The column degree
$\cdeg{\appbas}$ is defined similarly.  More generally, we will consider row
degrees shifted by some additive column weights: for a shift
$\shifts=(\shift{1},\ldots,\shift{\nbun})\in \shiftSpace$ the $\shifts$-row
degree of $\appbas$ is $\rdeg[\shifts]{\appbas} = (r_1,\ldots,r_\nbun)$ where
$r_i =
\max(\deg(\appbas_{i,1})+\shift{1},\ldots,\deg(\appbas_{i,\nbun})+\shift{\nbun})$.

We use $\sumVec{\cdot}$ to denote the sum of integer tuples: for example
$\sumVec{\rdeg[\shifts]{\appbas}}$ is the sum of the $\shifts$-row degree of
$\appbas$ (note that this sum might contain negative terms).  The comparison of
integer tuples is entrywise: $\cdeg{\sys} < \orders$ means that the column $j$
of $\sys$ has degree less than $\order_j$, for $1\le j\le \nbeq$.  When adding
a constant to a tuple, say for example $\shifts-1$, this stands for
the tuple $(\shift{1}-1,\ldots,\shift{\nbun}-1)$.

In existing approximant basis algorithms, the output bases may take different
forms: essentially, they can be $\shifts$-minimal (also called
$\shifts$-reduced \cite{BarBul92}), $\shifts$-weak Popov \cite{MulSto03}, or
$\shifts$-Popov \cite{BeLaVi99}.  For formal definitions and for motivating the
use of shifts, we direct the reader to these references and to those above
about approximant basis algorithms; here the precise form of the basis will not
play an important role.  What is however at the core of the efficiency of our
algorithms is the impact of these forms on the degrees in the basis.

In what follows, by \emph{size} of a matrix we mean the number of field
elements used for its dense representation.  We define the quantity
\[
  \size{\appbas}=\nbun^2 + \sum_{1\le i,j \le \nbun} \max(0,\deg(p_{ij}))
\]
for a matrix $\appbas = [p_{ij}] \in \appbasSpace$.  In the next paragraph, we
discuss degree bounds on $\appbas$ when it is the output of any of the
approximant basis algorithms mentioned above; note that these bounds all imply
that $\appbas$ has size in $\bigO{\nbun \vsdim}$.  

There is no general degree bound for approximant bases: any unimodular matrix
is a basis of $\modAppCustom{\orders}{\matz} = \appSpace$.  Still, a basis
$\appbas$ of $\modApp$ always satisfies $\deg(\det(\appbas)) \le \vsdim$.  Now,
for an $\shifts$-minimal $\appbas$, we have $\sumVec{\rdeg{\appbas}} \in
\bigO{\vsdim}$ as soon as $\sumVec{\shifts-\min(\shifts)} \in \bigO{\vsdim}$
\cite[Thm.\,4.1]{BarBul92}, and it was shown in \cite{ZhoLab12} that $\appbas$
has size in $\bigO{\nbun\vsdim}$ if $\sumVec{\!\max(\shifts)-\shifts} \in
\bigO{\vsdim}$.  Yet, without such assumptions on the shift, there are
$\shifts$-minimal bases whose size is in $\Theta(\nbun^2 \vsdim)$
\cite[App.\,B]{JeNeScVi16}, ruling out the feasibility of finding them in time
$\softO{\nbun^{\expmatmul-1} \vsdim}$.  In this case, the fastest known
algorithms return the more constrained $\shifts$-Popov basis $\appbas$, for
which $\sumVec{\cdeg{\appbas}} \le \vsdim$ holds independently of $\shifts$.

\myparagraph{Problem and contribution.}

Certifying that a matrix $\appbas$ is an $\shifts$-minimal approximant basis
for a given instance $(\orders,\sys,\shifts)$ boils down to the following three
properties of $\appbas$:
\begin{enumerate}[(1)]
  \item \emph{Minimal:} $\appbas$ is in $\shifts$-reduced
    form.  By definition, this amounts to testing the invertibility of the
    so-called $\shifts$-leading matrix of $\appbas$ (see Step\,\textbf{1}
    of \cref{algo:certif_appbas} for the construction of this matrix), which
    can be done using $\bigO{\nbun^\expmatmul}$ operations in $\field$.
  \item \emph{Approximant:} the rows of $\appbas$ are approximants.  That is,
    we should check that $\appbas \sys = \matz \bmod \mods$.  The difficulty is
    to avoid computing the full truncated product $\appbas\sys \bmod \mods$,
    since this costs $\softO{\nbun^{\expmatmul-1}\vsdim}$.  In
    \cref{sec:verif_truncmatprod}, we give a probabilistic algorithm which
    verifies more generally $\appbas \sys = \rhs \bmod \mods$ using
    $\bigO{\size{\appbas}+\nbun\vsdim}$ operations, without requiring a
    certificate.
  \item \emph{Basis:} the rows of $\appbas$ generate the
    approximant module\footnote{This is not implied by $(1)$ and $(2)$: for
      $\order = \max(\orders)$, then $\var^\order \idMat[\nbun]$ is
      $\shifts$-reduced and $\var^\order \idMat[\nbun] \sys = \matz \bmod
      \mods$ holds; yet, $\var^\order \idMat[\nbun]$ is not a basis of
    $\modApp$ for most $(\sys,\orders)$.}.  For this, we prove that it suffices
    to verify first that $\det(\appbas)$ is of the form $c \var^\delta$ for
    some $c\in\field\setminus\{0\}$ and where $\delta=\sumVec{\rdeg{\appbas}}$,
    and second that some constant $\nbun \times (\nbun+\nbeq)$ matrix has full
    rank; this matrix involves $\appbas(0)$ and the coefficient $\cmat$ of
    degree~$0$ of $\appbas\sys\xDiag{-\orders}$.  In \cref{sec:certif_appbas},
    we show that $\cmat$ can serve as a certificate, and that a probabilistic
    algorithm can assess its correctness at a suitable cost.
\end{enumerate}

Our (non-interactive) certification protocol is as follows.  Given
$(\orders,\sys,\shifts)$, the \emph{Prover} computes a matrix $\appbas$,
supposedly an $\shifts$-minimal basis of $\modApp$, along with a constant
matrix $\cmat \in \cmatSpace$, supposedly the coefficient of degree~$0$ of the
product $\appbas\sys\xDiag{-\orders}$. Then, the \emph{Prover} communicates
these results to the \emph{Verifier} who must solve \cref{pbm:certif_appbas}
within a cost asymptotically better than $\softO{\nbun^{\expmatmul-1}\vsdim}$.

\begin{problembox}
  \problemInfo
  {Approximant basis certification}
  {pbm:certif_appbas}

  \dataInfos{Input}{
    \item order $\orders \in \orderSpace$,
    \item matrix $\sys \in \sysSpace$ with $\cdeg{\sys} < \orders$,
    \item shift $\shifts \in \shiftSpace$,
    \item matrix $\appbas \in \appbasSpace$,
    \item certificate matrix $\cmat \in \cmatSpace$.
  }

  \dataInfo{Output}{
    \algoword{True} if $\appbas$ is an $\shifts$-minimal
    basis of $\modApp$ and $\cmat$ is the coefficient of degree $0$ of
    $\appbas\sys\xDiag{-\orders}$, otherwise \algoword{False}.
  }
\end{problembox}

The main result in this paper is an efficient solution to
\cref{pbm:certif_appbas}.

\begin{theorem}
  \label{thm:certif_appbas}
  There is a Monte-Carlo algorithm which solves \cref{pbm:certif_appbas} using
  $\bigO{\nbun \vsdim + \nbun^{\expmatmul-1}(\nbun+\nbeq)}$ operations in
  $\field$, assuming $\size{\appbas} \in \bigO{\nbun\vsdim}$.  It chooses
  $\nbun+2$ elements uniformly and independently at random from a finite subset
  $\fsubset \subset \field$.  If $\fsubset$ has cardinality at least
  $2(\vsdim+1)$, then the probability that a \algoword{True} answer is
  incorrect is less than $1/2$, while a \algoword{False} answer is always
  correct.
\end{theorem}
A detailed cost bound showing the constant factors is described in
\cref{prop:algo:certif_appbas}.  If $\size{\appbas} \in \bigO{\nbun\vsdim}$,
then the size of the input of \cref{pbm:certif_appbas} is in
$\bigO{\nbun\vsdim}$; the cost bound above is therefore optimal (up to constant
factors) as soon as $\nbun^{\expmatmul-2}(\nbun+\nbeq) \in \bigO{\vsdim}$.

If $\field$ is a small finite field, there may be no subset
$\fsubset\subset\field$ of cardinality $\#S \ge 2(\vsdim+1)$.  Then, our
approach still works by performing the probabilistic part of the computation
over a sufficiently large extension of $\field$.  Note that an extension of
degree about $1+\lceil\log_2(\vsdim)\rceil$ would be suitable; this would
increase our complexity estimates by a factor logarithmic in $\vsdim$, which
remains acceptable in our context.

Our second result is the efficient computation of the certificate.
\begin{theorem}
  \label{thm:certif_comp}
  Let $\orders\in\orderSpace$, let $\sys\in \sysSpace$ with
  $\cdeg{\sys}<\orders$ and $\nbun\in\bigO{\vsdim}$, and let $\appbas \in
  \appbasSpace$.  If $\sumVec{\rdeg{\appbas}} \in \bigO{\vsdim}$ or
  $\sumVec{\cdeg{\appbas}} \in \bigO{\vsdim}$, there is a deterministic
  algorithm which computes the coefficient of degree $0$ of $\appbas\sys
  \xDiag{-\orders}$ using $\bigO{\nbun^{\omega-1}\vsdim}$ operations
  in~$\field$ if $\nbun\geq\nbeq$ and $\bigO{\nbun^{\omega-1}\vsdim
  \log(\nbeq/\nbun)}$ operations in $\field$ if $\nbun<\nbeq$.
\end{theorem}

Note that the assumption $\nbun \in \bigO{\vsdim}$ in this theorem is commonly
made in approximant basis algorithms, since when $\vsdim \le \nbun$ most
entries of a minimal approximant basis have degree in $\bigO{1}$ and the
algorithms then rely on methods from dense $\field$-linear algebra.

\section{Certifying approximant bases}
\label{sec:certif_appbas}

Here, we present our certification algorithm. Its properties, given in
\cref{prop:algo:certif_appbas}, prove \cref{thm:certif_appbas}.  One of its
core components is the verification of truncated polynomial matrix products;
the details of this are in \cref{sec:verif_truncmatprod} and are taken for
granted here.

First, we show the basic properties behind the correctness of this algorithm,
which are summarized in the following result.

\begin{theorem}
  \label{thm:characterization_appbas}
  Let $\orders \in \orderSpace$, let $\sys \in \sysSpace$, and let
  $\shifts\in\shiftSpace$.  A matrix $\appbas \in \appbasSpace$ is an
  $\shifts$-minimal basis of $\modApp$ if and only if the following properties
  are all satisfied:
  \begin{enumerate}[(i)]
    \item \label{it:reduced} $\appbas$ is $\shifts$-reduced;
    \item \label{it:determinant} $\det(\appbas)$ is a nonzero monomial in
      $\polRing$;
    \item \label{it:approx} the rows of $\appbas$ are in $\modApp$, that is,
      $\appbas \sys = \matz \bmod \mods$;
    \item \label{it:kerbas} $[\appbas(0) \;\; \cmat] \in
      \matSpace[\nbun][(\nbun+\nbeq)]$ has full rank, where $\cmat$ is the
      coefficient of degree $0$ of $\appbas\sys\xDiag{-\orders}$.
  \end{enumerate}
\end{theorem}

We remark that having both $\appbas\sys = \matz \bmod \mods$ and $\cmat$ the
constant coefficient of $\appbas\sys\xDiag{-\orders}$ is equivalent to the
single truncated identity $\appbas \sys = \cmat \mods \bmod
\xDiag{\truncOrds}$, where $\truncOrds = (\order_1+1,\ldots,\order_\nbeq+1)$.

As mentioned above, the details of the certification of the latter identity is
deferred to \cref{sec:verif_truncmatprod}, where we present more generally the
certification for truncated products of the form $\appbas \sys = \rhs \bmod
\xDiag{\truncOrds}$.

Concerning \cref{it:determinant}, the fact that the determinant of any basis of
$\modApp$ must divide $\var^\vsdim$, where $\vsdim=\sumVec{\orders}$, is
well-known; we refer to \cite[Sec.\,2]{BecLab97} for a more general result.

The combination of \cref{it:reduced,it:approx} describes the set of matrices
$\appbas \in \appbasSpace$ which are $\shifts$-reduced and whose rows are in
$\modApp$.  For $\appbas$ to be an $\shifts$-minimal basis of $\modApp$, its
rows should further form a generating set for $\modApp$; thus, our goal here is
to prove that this property is realized by the combination of
\cref{it:determinant,it:kerbas}.

For this, we will rely on a link between approximant bases and kernel bases,
given in \cref{lem:appbas_kerbas}.  We recall that, for a given matrix $\mat{M}
\in \polMatSpace[\mu][\nu]$ of rank $r$,
\begin{itemize}
  \item a \emph{kernel basis} for $\mat{M}$ is a matrix in
    $\polMatSpace[(\mu-r)][\mu]$ whose rows form a basis of the left kernel
    $\{\app \in \polMatSpace[1][\mu] \mid \app\mat{M} = \matz\}$,
  \item a \emph{column basis} for $\mat{M}$ is a matrix in
    $\polMatSpace[\mu][r]$ whose columns form a basis of the column space
    $\{\mat{M}\app, \app\in\polMatSpace[\nu][1]\}$.
\end{itemize}

In particular, by definition, a kernel basis has full row rank and a column
basis has full column rank.  The next result states that the column space of a
kernel basis is the whole space (that is, the space spanned by the identity
matrix).

\begin{lemma}
  \label{lem:kerbas_unimodular_colbas} 
  Let $\mat{M} \in \polMatSpace[\mu][\nu]$ and let $\mat{B} \in
  \polMatSpace[k][\mu]$ be a kernel basis for $\mat{M}$.  Then, any column
  basis for $\mat{B}$ is unimodular.  Equivalently, $\mat{B} \mat{U} =
  \idMat[k]$ for some $\mat{U} \in \polMatSpace[\mu][k]$.
\end{lemma}
\begin{proof}
  Let $\mat{S} \in \polMatSpace[k]$ be a column basis for $\mat{B}$.  By
  definition, $\mat{B} = \mat{S} \matt{B}$ for some $\matt{B} \in
  \polMatSpace[k][\mu]$.  Then $\matz = \mat{B}\mat{M} =
  \mat{S}\matt{B}\mat{M}$, hence $\matt{B}\mat{M} = \matz$ since $\mat{S}$ is
  nonsingular.  Thus, $\mat{B}$ being a kernel basis for $\mat{M}$, we have
  $\matt{B} = \mat{T} \mat{B}$ for some $\mat{T} \in \polMatSpace[k][k]$.  We
  obtain $(\mat{S}\mat{T} - \idMat[k]) \mat{B} = \matz$, hence $\mat{S}\mat{T}
  = \idMat[k]$ since $\mat{B}$ has full row rank.  Thus, $\mat{S}$ is
  unimodular.
\end{proof}

This arises for example in the computation of column bases and unimodular
completions in \cite{ZhoLab13,ZhoLab14}; the previous lemma can also be derived
from these references, and in particular from \cite[Lem.\,3.1]{ZhoLab13}.

Here, we will use the property of \cref{lem:kerbas_unimodular_colbas} for a
specific kernel basis, built from an approximant basis as follows.

\begin{lemma}
  \label{lem:appbas_kerbas}
  Let $\orders \in \orderSpace$, $\sys \in \sysSpace$, and $\appbas \in
  \appbasSpace$.  Then, $\appbas$ is a basis of $\modApp$ if and only if there
  exists $\mat{Q} \in \sysSpace$ such that $[\appbas \;\; \mat{Q}]$ is a kernel
  basis for $\trsp{[\trsp{\sys} \;\; -\xDiag{\orders}]}$.  If this is the case,
  then we have $\mat{Q} = \appbas \sys \xDiag{-\orders}$ and there exist
  $\mat{V} \in \polMatSpace[\nbun]$ and $\mat{W} \in
  \polMatSpace[\nbeq][\nbun]$ such that $\appbas \mat{V} + \mat{Q} \mat{W} =
  \idMat[\nbun]$.
\end{lemma}
\begin{proof}
  The equivalence is straightforward; a detailed proof can be found in
  \cite[Lem.\,8.2]{Neiger16b}.  If $[\appbas \;\; \mat{Q}]$ is a kernel basis
  for $\trsp{[\trsp{\sys} \;\; -\xDiag{\orders}]}$, then we have $\appbas \sys
  = \mat{Q} \mods$, hence the explicit formula for $\mat{Q}$.  Besides, the
  last claim is a direct consequence of \cref{lem:kerbas_unimodular_colbas}.
\end{proof}

This leads us to the following result, which forms the main ingredient
that was missing in order to prove \cref{thm:characterization_appbas}.

\begin{lemma}
  \label{lem:certif_via_rank}
  Let $\orders \in \orderSpace$ and let $\sys \in \sysSpace$.  Let $\appbas \in
  \appbasSpace$ be such that $\appbas \sys = \matz \bmod \mods$ and $\det(\appbas)$
  is a nonzero monomial, and let $\cmat \in \matSpace[\nbun][(\nbun+\nbeq)]$
  be the constant coefficient of $\appbas \sys \xDiag{-\orders}$.  Then,
  $\appbas$ is a basis of $\modApp$ if and only if $[\appbas(0) \;\; \cmat]
  \in \matSpace[\nbun][(\nbun+\nbeq)]$ has full rank.
\end{lemma}
\begin{proof}
  First, assume that $\appbas$ is a basis of $\modApp$.  Then, defining
  $\mat{Q} = \appbas \sys \xDiag{-\orders} \in \sysSpace$,
  \cref{lem:appbas_kerbas} implies that $\appbas \mat{V} + \mat{Q} \mat{W} =
  \idMat[\nbun]$ for some $\mat{V} \in \polMatSpace[\nbun]$ and $\mat{W} \in
  \polMatSpace[\nbeq][\nbun]$.  Since $\mat{Q}(0) = \cmat$, this yields
  $\appbas(0) \mat{V}(0) + \cmat \mat{W}(0) = \idMat[\nbun]$, and thus
  $[\appbas(0) \;\; \cmat]$ has full rank.

  Now, assume that $\appbas$ is not a basis of $\modApp$.  If $\appbas$ has
  rank $<\nbun$, then $[\appbas(0) \;\; \cmat]$ has rank $<\nbun$ as well.
  If $\appbas$ is nonsingular, $\appbas = \mat{U} \mat{A}$ for some basis
  $\mat{A}$ of $\modApp$ and some $\mat{U}\in\appbasSpace$ which is nonsingular
  but not unimodular.  Then, $\det(\mat{U})$ is a nonconstant divisor of the
  nonzero monomial $\det(\appbas)$; hence $\det(\mat{U})(0) = 0 =
  \det(\mat{U}(0))$, and thus $\mat{U}(0)$ has rank $<\nbun$.  Since $[\appbas
  \;\; \mat{Q}] = \mat{U} [\mat{A} \;\; \mat{A}\sys\xDiag{-\orders}]$, it
  directly follows that $[\appbas(0) \;\; \cmat]$ has rank $<\nbun$.
\end{proof}

\begin{proof}[Proof of \cref{thm:characterization_appbas}]
  If $\appbas$ is an $\shifts$-minimal basis of $\modApp$, then by definition
  \cref{it:reduced,it:approx} are satisfied.  Since the rows of
  $\var^{\max(\orders)}\idMat[\nbun]$ are in $\modApp$ and $\appbas$ is a
  basis, the matrix $\var^{\max(\orders)}\idMat[\nbun]$ is a left multiple of
  $\appbas$ and therefore the determinant of $\appbas$ divides
  $\var^{\nbun\max(\orders)}$: it is a nonzero monomial.  Then, according to
  \cref{lem:certif_via_rank}, $[\appbas(0) \;\; \cmat]$ has full rank.
  Conversely, if \cref{it:determinant,it:approx,it:kerbas} are satisfied, then
  \cref{lem:certif_via_rank} states that $\appbas$ is a basis of $\modApp$;
  thus if furthermore \cref{it:reduced} is satisfied then $\appbas$ is an
  $\shifts$-minimal basis of $\modApp$.
\end{proof}

\begin{algobox}
  \algoInfo
  {CertifApproxBasis}
  {algo:certif_appbas}

  \dataInfos{Input}{
    \item order $\orders = (\order_1,\ldots,\order_\nbeq) \in \orderSpace$,
    \item matrix $\sys \in \sysSpace$ with $\cdeg{\sys} < \orders$,
    \item shift $\shifts = (\shift{1},\ldots,\shift{\nbun}) \in \shiftSpace$,
    \item matrix $\appbas \in \appbasSpace$,
    \item certificate matrix $\cmat \in \cmatSpace$.
  }

  \dataInfo{Output}{
    \algoword{True} if $\appbas$ is an $\shifts$-minimal basis of $\modApp$ and
    $\cmat$ is the constant term of $\appbas\sys\xDiag{-\orders}$, otherwise
    \algoword{True} or \algoword{False}.
  }

  \algoSteps{
    \item \inlcomment{$\appbas$ not in $\shifts$-reduced form $\Rightarrow$ False} \\
      $\mat{L} \assign$ the matrix in $\matSpace[\nbun]$ whose entry $i,j$ is
      the coefficient of degree $\rdeg[\shifts]{\matrow{\appbas}{i}} -
      \shift{j}$ of the entry $i,j$ of $\appbas$ \\
      \algoword{If} $\mat{L}$ is not invertible \algoword{then} \algoword{return} \algoword{False}
    \item \inlcomment{$\rank{[\appbas(0) \;\; \cmat]}$ not full rank $\Rightarrow$ False} \\
      \algoword{If} $\rank{[\appbas(0) \;\; \cmat]} < \nbun$ \algoword{then} \algoword{return} \algoword{False}
    \item \inlcomment{$\det(\appbas)$ not a nonzero monomial $\Rightarrow$ False} \\
      $\fsubset \assign$ a finite subset of $\field$ \\
      $\Delta \assign \sumVec{\rdeg[\shifts]{\appbas}}-\sumVec{\shifts}$ \\
      $\evpt \assign$ chosen uniformly at random from $\fsubset$ \\
      \algoword{If} $\det(\appbas(\evpt)) \neq \det(\appbas(1)) \evpt^\Delta$ \algoword{then} \algoword{return} \algoword{False}
    \item \inlcomment{certify truncated product $\appbas\sys = \cmat\mods \bmod \xDiag{\truncOrds}$} \\
      $\truncOrds \assign (\order_1+1,\ldots,\order_\nbeq+1)$ \\
      \algoword{Return} $\algoname{VerifTruncMatProd}(\truncOrds,\appbas,\sys,\cmat\mods)$
  }
\end{algobox}

In order to provide a sharp estimate of the cost of \cref{algo:certif_appbas},
we recall the best known cost bound with constant factors of the LQUP
factorization of an $\nbun\times\nbeq$ matrix over $\field$, which we use for
computing ranks and determinants. Assuming $\nbun\le\nbeq$, we have:
\[
  \costLQUP{\nbun}{\nbeq} =
  \left(\left\lceil \frac{\nbeq}{\nbun} \right\rceil \frac{1}{2^{\expmatmul-1}-2} - \frac{1}{2^\expmatmul-2}\right) MM(\nbun)
\]
operations in $\field$ \cite[Lem.\,5.1]{DumGioPer08}, where $MM(\nbun)$ is the
cost for the multiplication of $\nbun\times\nbun$ matrices over $\field$.

\begin{proposition}
  \label{prop:algo:certif_appbas}
  \cref{algo:certif_appbas} uses at most
  \begin{align*}
    & 5\size{\appbas} + 2\nbun (\vsdim + \max(\orders)) + 3\costLQUP{\nbun}{\nbun}
    + \costLQUP{\nbun}{\nbun+\nbeq}  \\
    & \qquad\qquad\qquad\qquad\;\; + (4\nbun+1)\nbeq + 4\log_2(\vsdim\order_1\cdots\order_\nbeq)  \\
    & \in\bigO{\size{\appbas} + \nbun\vsdim + \nbun^{\expmatmul-1} (\nbun+\nbeq)}
  \end{align*}
  operations in $\field$, where $\vsdim=\sumVec{\orders}$.  It is a
  false-biased Monte Carlo algorithm.  If $\appbas$ is not an $\shifts$-minimal
  basis of $\modApp$, then the probability that it outputs \algoword{True} is
  less than $\frac{\vsdim+1}{\#\fsubset}$, where $\fsubset$ is the
  finite subset of $\field$ from which random field elements are drawn.
\end{proposition}
\begin{proof}
  By definition, $\appbas$ is $\shifts$-reduced if and only if its
  $\shifts$-leading matrix $\mat{L}$ computed at Step\,\textbf{1} is
  invertible.  Thus, Step\,\textbf{1} correctly tests the property in
  \cref{it:reduced} of \cref{thm:characterization_appbas}.  It uses at most
  $\costLQUP{\nbun}{\nbun}$ operations in $\field$.  Furthermore,
  Step\,\textbf{2} correctly tests the first part of \cref{it:kerbas} of
  \cref{thm:characterization_appbas} and uses at most
  $\costLQUP{\nbun}{\nbun+\nbeq}$ operations.

  Step\,\textbf{3} performs a false-biased Monte Carlo verification of
  \cref{it:determinant} of \cref{thm:characterization_appbas}.  Indeed, since
  $\appbas$ is $\shifts$-reduced (otherwise the algorithm would have exited at
  Step~\textbf{1}), we know from \cite[Sec.\,6.3.2]{Kailath80} that
  $
    \deg(\det(\appbas)) = \Delta =
    \sumVec{\rdeg[\shifts]{\appbas}}-\sumVec{\shifts}
  $.
  Thus, $\det(\appbas)$ is
  a nonzero monomial if and only if $\det(\appbas) = \det(\appbas(1))
  \var^\Delta$.  Step\,\textbf{3} tests the latter equality by evaluation at a
  random point $\evpt$.  The algorithm only returns \algoword{False} if
  $\det(\appbas(\evpt)) \neq \det(\appbas(1)) \evpt^\Delta$, in which case
  $\det(\appbas)$ is indeed not a nonzero monomial.  Furthermore, if we have
  $\det(\appbas) \neq \det(\appbas(1)) \var^\Delta$, then the probability that
  the algorithm fails to detect this, meaning that $\det(\appbas(\evpt)) =
  \det(\appbas(1)) \evpt^\Delta$, is at most $\frac{\Delta}{\#\fsubset}$.
  Since $\Delta \le \vsdim$ according to \cite[Thm.\,4.1]{BarBul92}, this is
  also at most $\frac{\vsdim}{\#\fsubset} < \frac{\vsdim+1}{\#\fsubset}$.  

  The evaluations $\appbas(\evpt)$ and $\appbas(1)$ are computed using
  respectively at most $2(\size{\appbas}-\nbun^2)$ operations and at most
  $\size{\appbas}-\nbun^2$ additions.  Then, computing the two determinants
  $\det(\appbas(\evpt))$ and $\det(\appbas(1))$ uses at most
  $2\costLQUP{\nbun}{\nbun}+2\nbun$ operations.  Finally, computing
  $\det(\appbas(1)) \evpt^\Delta$ uses at most $2\log_2(\Delta)+1 \le
  2\log_2(\vsdim)+1$ operations.

  Summing the cost bounds for the first three steps gives
  \begin{align}
    & 3(\size{\appbas}-\nbun^2) + 3\costLQUP{\nbun}{\nbun} + \costLQUP{\nbun}{\nbun+\nbeq} + 2\nbun + 2 \log_2(\vsdim)+1 \nonumber \\
    & \le 3\size{\appbas} + 3\costLQUP{\nbun}{\nbun} + \costLQUP{\nbun}{\nbun+\nbeq} + 2 \log_2(\vsdim).
    \label{eqn:cost_certif_except_product}
  \end{align}

  Step\,\textbf{4} tests the identity $\appbas\sys = \cmat\mods \bmod
  \xDiag{\truncOrds}$, which corresponds to both \cref{it:approx} of
  \cref{thm:characterization_appbas} and the second part of \cref{it:kerbas}.
  \cref{prop:algo:verif_truncmatprod} ensures that:
  \begin{itemize}
    \item If the call to \algoname{VerifTruncMatProd} returns \algoword{False},
      we have $\appbas \sys \neq \cmat \mods \bmod \xDiag{\truncOrds}$, and
      \cref{algo:certif_appbas} correctly returns \algoword{False}.
    \item If $\appbas \sys \neq \cmat \mods \bmod \xDiag{\truncOrds}$ holds,
      the probability that \cref{algo:certif_appbas} fails to detect this (that
      is, the call at Step\,\textbf{4} returns \algoword{True}) is less than
      $\frac{\max(\orders)+1}{\#\fsubset}$.
  \end{itemize}
  A cost bound for Step\,\textbf{4} is given in
  \cref{prop:algo:verif_truncmatprod}, with a minor improvement for the present
  case given in \cref{rmk:specific_rhs}.  Summing it with the bound in
  \cref{eqn:cost_certif_except_product} gives a cost bound for
  \cref{algo:certif_appbas}, which is bounded from above by that in the
  proposition.

  Thanks to \cref{thm:characterization_appbas}, the above considerations show
  that when the algorithm returns \algoword{False}, then $\appbas$ is indeed
  not an $\shifts$-minimal basis of $\modApp$.  On the other hand, if $\appbas$
  is not an $\shifts$-minimal basis of $\modApp$, the algorithm returns
  \algoword{True} if and only if one of the probabilistic verifications in
  Steps\,\textbf{3} and\,\textbf{4} take the wrong decision.  According to
  the probabilities given above, this may happen with probability less than
  $\max(\frac{\vsdim+1}{\#\fsubset},\frac{\max(\orders)+1}{\#\fsubset}) =
  \frac{\vsdim+1}{\#\fsubset}$. 
\end{proof}

\section{Verifying a truncated product}
\label{sec:verif_truncmatprod}

In this section, we focus on the verification of truncated products of
polynomial matrices, and we give the corresponding algorithm
\algoname{VerifTruncMatProd} used in \cref{algo:certif_appbas}.

Given a truncation order $\truncOrds$ and polynomial matrices $\appbas$,
$\sys$, $\rhs$, our goal is to verify that $\appbas \sys = \rhs \bmod
\xDiag{\truncOrds}$ holds with good probability.  Without loss of generality,
we assume that the columns of $\sys$ and $\rhs$ are already truncated with
respect to the order $\truncOrds$, that is, $\cdeg{\sys} < \truncOrds$ and
$\cdeg{\rhs} < \truncOrds$.  Similarly, we assume that $\appbas$ is truncated
with respect to $\dmax = \max(\truncOrds)$, that is, $\deg(\appbas) < \dmax$.

\begin{problembox}
  \problemInfo
  {Truncated matrix product verification}
  {pbm:verif_truncmatprod}

  \dataInfos{Input}{
    \item truncation order $\truncOrds \in \truncOrdSpace$,
    \item matrix $\appbas \in \appbasSpace$ with $\deg(\appbas) < \max(\truncOrds)$,
    \item matrix $\sys \in \sysSpace$ with $\cdeg{\sys} < \truncOrds$,
    \item matrix $\rhs \in \sysSpace$ with $\cdeg{\rhs} < \truncOrds$.
  }

  \dataInfo{Output}{
    \algoword{True} if $\appbas \sys = \rhs \bmod \xDiag{\truncOrds}$,
    otherwise \algoword{False}.
  }
\end{problembox}

Obviously, our aim is to obtain a verification algorithm which has a
significantly better cost than the straightforward approach which computes the
truncated product $\appbas\sys \bmod \xDiag{\truncOrds}$ and compares it with
the matrix $\rhs$.  To take an example: if we have $\nbeq\in\bigO{\nbun}$ as
well as $\sumVec{\rdeg{\appbas}} \in \bigO{\sumVec{\truncOrds}}$ or
$\sumVec{\cdeg{\appbas}} \in \bigO{\sumVec{\truncOrds}}$, as commonly happens
in approximant basis computations, then this truncated product $\appbas\sys
\bmod \xDiag{\truncOrds}$ can be computed using $\softO{\nbun^{\expmatmul-1}
\sumVec{\truncOrds}}$ operations in $\field$.

For verifying the non-truncated product $\appbas\sys = \rhs$, the classical
approach would be to use evaluation at a random point, following ideas from
\cite{Schwartz80,Zippel79,DeMilloLipton78}.  However, evaluation does not
behave well with regards to truncation.  A similar issue was tackled in
\cite{Gi17} for the verification of the middle product and the short products
of univariate polynomials.  The algorithm of \cite{Gi17} can be adapted to work
with polynomial matrices by writing them as univariate polynomials with matrix
coefficents; for example, $\appbas$ is a polynomial $\appbas=\sum_{0\le i
<\dmax} \appbas_i\var^i$ with coefficients $\appbas_i\in\matSpace[\nbun]$.
While this leads to a verification of $\appbas \sys = \rhs \bmod
\xDiag{\truncOrds}$ with a good probability of success, it has a cost which is
close to that of computing $\appbas \sys \bmod \xDiag{\truncOrds}$.

To lower down the cost, we will combine the evaluation of truncated products
from \cite{Gi17} with Freivalds' technique \cite{Freivalds1979}.  The latter
consists in left-multiplying the matrices by some random vector $\uvec\in
\uvecSpace$, and rather checking whether $\uvec\appbas\sys = \uvec\rhs \bmod
\xDiag{\truncOrds}$; this effectively reduces the row dimension of the
manipulated matrices, leading to faster computations.  Furthermore, this does
not harm the probability of success of the verification, as we detail now.

In what follows, given a matrix $\mat{A}\in\sysSpace$ and an order
$\truncOrds\in\truncOrdSpace$, we write $\mat{A} \rem \xDiag{\truncOrds}$ for
the (unique) matrix $\mat{B} \in \sysSpace$ such that $\mat{B} = \mat{A} \bmod
\xDiag{\truncOrds}$ and $\cdeg{\mat{B}} < \truncOrds$.  For simplicity, we will
often write $\mat{A}_1\mat{A}_2 \rem \xDiag{\truncOrds}$ to actually mean
$(\mat{A}_1 \mat{A}_2) \rem \xDiag{\truncOrds}$.

\begin{lemma}
  \label{lem:proba_failure}
  Let $\fsubset$ be a finite subset of $\field$.  Let $\uvec\in\uvecSpace$ with
  entries chosen uniformly and independently at random from $\fsubset$, and let
  $\evpt\in\field$ be chosen uniformly at random from $\fsubset$.  Assuming
  $\appbas \sys \neq \rhs \bmod \xDiag{\truncOrds}$, the probability that
  $(\uvec\appbas\sys \rem \xDiag{\truncOrds})(\evpt)=\uvec\rhs(\evpt)$ is
  less than $\frac{\max(\truncOrds)}{\#\fsubset}$.
\end{lemma}
\begin{proof}
  Let $\mat{A} = (\appbas \sys - \rhs) \rem \xDiag{\truncOrds}$.  By
  assumption, there exists a pair $(i,j)$ such that the entry $(i,j)$ of
  $\mat{A}$ is nonzero.  Since this entry is a polynomial in $\polRing$ of
  degree less than $\dmax=\max(\truncOrds)$, the probability that $\evpt$ is a
  root of this entry is at most $\frac{\dmax-1}{\#\fsubset}$.  As a
  consequence, we have $\mat{A}(\evpt) \neq \matz \in \matSpace[\nbun][\nbeq]$
  with probability at least $1-\frac{\dmax-1}{\#\fsubset}$.  In this case,
  $\uvec \mat{A}(\evpt) = \matz$ occurs with probability at most
  $\frac{1}{\#\fsubset}$ (see \cite[Sec.\,7.1]{MotRag95}).
  
  Thus, altogether the probability that $\uvec \mat{A}(\evpt) = \matz$ is
  bounded from above by
  $
    \frac{\dmax-1}{\#\fsubset} +
    \left(1-\frac{\dmax-1}{\#\fsubset}\right)\frac{1}{\#\fsubset} <
    \frac{\dmax}{\#\fsubset}
  $,
  which concludes the proof.
\end{proof}
 
We deduce an approach to verify the truncated product: compute $\uvec
\mat{A}(\evpt) = ((\uvec \appbas \sys - \uvec \rhs) \rem
\xDiag{\truncOrds})(\evpt)$ and check whether it is zero or nonzero.  The
remaining difficulty is to compute $\uvec \mat{A}(\evpt)$ efficiently: we will
see that this can be done in $O(\size{\appbas} + \nbun \sumOrds)$ operations.

For this, we use a strategy similar to that in \cite[Lem.\,4.1]{Gi17} and
essentially based on the following formula for the truncated product.  Consider
a positive integer $\truncOrd \le \dmax$ and a vector $\col{f} \in
\polMatSpace[\nbun][1]$ of degree less than $\truncOrd$; one may think of
$\col{f}$ as a column $\matcol{\sys}{j}$ of $\sys$ and of $\truncOrd$ as the
corresponding order $\truncOrd_j$.  Writing $\col{f} = \sum_{0\le k<\truncOrd}
\col{f}_k \var^k$ with $\col{f}_k \in \matSpace[\nbun][1]$ and $\uvec\appbas =
\sum_{0\le k<\dmax} \app_k \var^k$ with $\app_k \in \matSpace[1][\nbun]$, we
have
\begin{align*}
  \uvec\appbas\col{f} \rem \var^\truncOrd
   & = \sum_{k=0}^{\truncOrd-1} \left(\sum_{i=0}^{\truncOrd-1-k} \app_i
       \var^{i}\right) \col{f}_{k} \var^{k}  \\
   & = \var^{\truncOrd-1}\sum_{k=0}^{\truncOrd-1} \left(\sum_{i=0}^{\truncOrd-1-k} \app_{\truncOrd-1-k-i} \var^{-i}\right) \col{f}_{k}  .
\end{align*}

Thus, the evaluation can be expressed as
\begin{equation}
  \label{eqn:eval_truncprod}
  (\uvec\appbas\col{f} \rem \var^\truncOrd)(\evpt)
  = \evpt^{\truncOrd-1}\sum_{k=0}^{\truncOrd-1} \hvec{\truncOrd-1-k} \col{f}_{k} ,
\end{equation}
where we define, for $0\le k < \dmax$,
\begin{equation}
  \label{eqn:hvec}
  \hvec{k} =  (\uvec \appbas \rem \var^{k+1})(\evpt^{-1}) = \sum_{i=0}^k \app_{k-i}
  \evpt^{-i} \in \matSpace[1][\nbun].
\end{equation}

These identities give an algorithm to compute the truncated product evaluation
$(\uvec\appbas\col{f} \rem \var^\truncOrd)(\evpt)$, which we sketch as follows:
\begin{itemize}
  \item apply Horner's method to the reversal of
    $\uvec\appbas\rem\var^{\truncOrd}$ at the point $\evpt^{-1}$, storing the
    intermediate results which are exactly the $\truncOrd$ vectors
    $\hvec{0},\ldots,\hvec{\truncOrd-1}$;
  \item compute the scalar products $\lambda_k = \hvec{\truncOrd-1-k}\col{f}_k$
    for $0\le k < \truncOrd$;
  \item compute $\evpt^{\truncOrd-1}$ and then $\evpt^{\truncOrd-1}\sum_{0\le k<\truncOrd}
    \lambda_k$.
\end{itemize}

The last step gives the desired evaluation according to
\cref{eqn:eval_truncprod}.  In our case, this will be applied to each column
$\col{f} = \matcol{\sys}{j}$ for $1\le j\le\nbeq$.  We will perform the first
item only once to obtain the $\dmax$ vectors $\hvec{0},\ldots,\hvec{\dmax-1}$,
since they do not depend on $\col{f}$.

\begin{algobox}
  \algoInfo
  {VerifTruncMatProd}
  {algo:verif_truncmatprod}

  \dataInfos{Input}{
    \item truncation order $\truncOrds = (\truncOrd_1,\ldots,\truncOrd_\nbeq) \in \truncOrdSpace$,
    \item matrix $\appbas \in \appbasSpace$ such that $\deg(\appbas) < \dmax = \max(\truncOrds)$,
    \item matrix $\sys = [\rpol_{ij}] \in \sysSpace$ with $\cdeg{\sys} < \truncOrds$,
    \item matrix $\rhs \in \sysSpace$ with $\cdeg{\rhs} < \truncOrds$.
  }

  \dataInfo{Output}{
    \algoword{True} if $\appbas \sys = \rhs \bmod \xDiag{\truncOrds}$,
    otherwise \algoword{True} or \algoword{False}.
  }

  \algoSteps{
    \item \inlcomment{Main objects for verification} \\
      $\fsubset \assign$ a finite subset of $\field$ \\
      $\evpt \assign$ element of $\field$ chosen uniformly at random from $\fsubset$ \\
      $\row{u} \assign$ vector in $\matSpace[1][\nbun]$ with entries chosen
      uniformly and independently at random from $\fsubset$
    \item \inlcomment{Freivalds: row dimension becomes $1$} \\
      $\app \assign \uvec \appbas$ \eolcomment{in $\polMatSpace[1][\nbun]$, degree $<\dmax$} \\
      $\row{g} \assign \uvec \rhs$ \eolcomment{in $\polMatSpace[1][\nbeq]$, $\cdeg{\row{g}} < \truncOrds$}
    \item \inlcomment{Evaluation of right-hand side: $\uvec \rhs(\evpt)$} \\
      write $\row{g} = [g_1 \; \cdots \; g_\nbeq]$ with $g_j \in \polRing$ of degree $<\truncOrd_j$ \\
      \algoword{For} $j$ \algoword{from} $1$ \algoword{to} $\nbeq$: \\
      \phantom{For} $e_j \assign g_j(\evpt)$
    \item \inlcomment{Truncated evaluations $\hvec{0},\ldots,\hvec{\dmax-1}$} \\
      write $\app = \sum_{0\le k<\dmax} \app_k \var^k$ with $\app_k \in \matSpace[1][\nbun]$ \\
      $\hvec{0} \assign \app_{0}$ \\
      \algoword{For} $k$ \algoword{from} $1$ \algoword{to} $\dmax-1$: \\
      \phantom{For} $\hvec{k} \assign \app_k + \evpt^{-1} \hvec{k-1}$
    \item \inlcomment{Evaluation of left-hand side: $(\uvec\appbas\sys \rem \xDiag{\truncOrds})(\evpt)$} \\
      \algoword{For} $j$ \algoword{from} $1$ \algoword{to} $\nbeq$: \eolcomment{process column $\matcol{\sys}{j}$} \\
      \phantom{For} write $\matcol{\sys}{j} = \sum_{0\le k<\truncOrd_j} \col{f}_k \var^k$ \\
      \phantom{For} $(\lambda_k)_{0\le k<\truncOrd_j} \assign (\hvec{\truncOrd_j-1-k} \cdot \col{f}_k)_{0\le k<\truncOrd_j}$ \\
      \phantom{For} $e'_j \assign \evpt^{\truncOrd_j-1} \sum_{0\le k<\truncOrd_j} \lambda_k$
    \item \algoword{If} $e_j \neq e'_j$ for some $j \in \{1,\ldots,\nbeq\}$ \algoword{then} \algoword{return} \algoword{False} \\
      \algoword{Else} \algoword{return} \algoword{True}
  }
\end{algobox}

\begin{proposition}
  \label{prop:algo:verif_truncmatprod}
  \cref{algo:verif_truncmatprod} uses at most
  \[
    2 \size{\appbas} + (6\nbun+1) \sumOrds + 2\nbeq\log_2(\dmax)
    \in \bigO{\size{\appbas} + \nbun\sumOrds + \nbeq\log_2(\dmax)}
  \]
  operations in $\field$, where $\dmax \le \sumOrds$ is the largest of the
  truncation orders.  It is a false-biased Monte Carlo algorithm.  If $\appbas
  \sys \neq \rhs \bmod \xDiag{\truncOrds}$, the probability that it outputs
  \algoword{True} is less than $\frac{\dmax}{\#\fsubset}$, where $\fsubset$ is
  the finite subset of $\field$ from which random field elements are drawn.
\end{proposition}
\begin{proof}
  The discussion above shows that this algorithm correctly computes
  $[e_j]_{1\le j\le\nbeq} = \uvec\rhs(\evpt)$ and $[e'_j]_{1\le j\le\nbeq} =
  (\uvec\appbas\sys \rem \xDiag{\truncOrds})(\evpt)$. If it returns
  \algoword{False}, then there is at least one $j$ for which $e'_j \neq e_j$,
  thus we must have $\uvec\appbas\sys \rem \xDiag{\truncOrds} \neq \uvec\rhs$
  and therefore $\appbas\sys \neq \rhs \bmod \xDiag{\truncOrds}$.  Besides, the
  algorithm correctly returns \algoword{True} if $\appbas\sys = \rhs \bmod
  \xDiag{\truncOrds}$.

  The analysis of the probability of failure (the algorithm returns
  \algoword{True} while $\appbas\sys \neq \rhs \bmod \xDiag{\truncOrds}$) is a
  direct consequence of \cref{lem:proba_failure}.

  Step\,\textbf{2} uses at most $2\size{\appbas} + (2\nbun-1)\sumOrds$
  operations in $\field$.  The Horner evaluations at Steps\,\textbf{3}
  and\,\textbf{4} require at most $2(\sumOrds-\nbeq)$ and at most
  $1+2\nbun(\dmax-1)$ operations, respectively.  Now, we consider the $j$-th
  iteration of the loop at Step\,\textbf{5}.  The scalar products
  $(\lambda_k)_{0\le k<\truncOrd_j}$ are computed using at most
  $(2\nbun-1)\truncOrd_j$ operations; the sum and multiplication by
  $\evpt^{\truncOrd_j-1}$ giving $e'_j$ use at most $\truncOrd_j+
  2\log_2(\truncOrd_j-1)$ operations.  Summing over $1\le j\le \nbeq$, this
  gives a total of at most
  $2\nbun\sumOrds+2\log_2((\truncOrd_1-1)\cdots(\truncOrd_\nbeq-1))$ operations
  for Step\,\textbf{5}.  Finally, Step\,\textbf{6} uses at most $\nbeq$
  comparisons of two field elements.  Summing these bounds for each step yields
  the cost bound
  \begin{equation}
    \label{eqn:cost_bound_verif_truncmatprod}
    2\size{\appbas} + (4\nbun+1)\sumOrds + 2\nbun(\dmax-1) - \nbeq + 2\log_2((\truncOrd_1-1)\cdots(\truncOrd_\nbeq-1)),
  \end{equation}
  which is at most the quantity in the proposition.
\end{proof}

In the certification of approximant bases, we want to verify a truncated matrix
product in the specific case where each entry in the column $j$ of $\rhs$ is
simply zero or a monomial of degree $\truncOrd_j-1$.  Then, a slightly better
cost bound can be given, as follows.

\begin{remark}
  \label{rmk:specific_rhs}
  Assume that $\truncOrds = (\order_1+1,\ldots,\order_\nbeq+1)$ and $\rhs =
  \cmat \mods$, for some $\orders=(\order_1,\ldots,\order_\nbeq)\in\orderSpace$
  and some constant $\cmat\in\cmatSpace$.  Then, the computation of $\uvec\rhs$
  at Step\,\textbf{2} uses at most $(2\nbun-1)\nbeq$ operations in $\field$.
  Besides, since the polynomial $g_j$ at Step\,\textbf{3} is either zero or a
  monomial of degree $\order_j$, its evaluation $e_j$ is computed using at most
  $2\log_2(\order_j) + 1$ operations via repeated squaring
  \cite[Sec.\,4.3]{vzGathen13}.  Thus, Step\,\textbf{3} uses at most
  $2\log_2(\order_1\cdots\order_\nbeq)+\nbeq$ operations.  As a result,
  defining $\vsdim = \sumVec{\orders}$, the cost bound in
  \cref{eqn:cost_bound_verif_truncmatprod} is lowered to
  \begin{align*}
    & 2 \size{\appbas} + 2\nbun (\sumOrds + \dmax-1 + \nbeq) + \nbeq + 4\log_2(\order_1\cdots\order_\nbeq) + 1 \\
    = { } & 2 \size{\appbas} + 2\nbun (\vsdim + \max(\orders) + 2\nbeq) + \nbeq + 4\log_2(\order_1\cdots\order_\nbeq) + 1.
    \qed
  \end{align*}
\end{remark}

\section{Computing the certificate}
\label{sec:compute_certificate}

\subsection{Context}
\label{sec:compute_certificate:intro}

In this section, we show how to efficiently compute the certificate
$\cmat\in\cmatSpace$, which is the term of degree $0$ of the product $\appbas
\sys \xDiag{-\orders}$, whose entries are Laurent polynomials (they are in
$\polRing$ if and only if the rows of $\appbas$ are approximants).
Equivalently, the column $\matcol{\cmat}{j}$ is the term of degree $\order_j$
of the column $j$ of $\appbas \sys$, where $\orders =
(\order_1,\ldots,\order_\nbeq)$.

We recall the notation $\vsdim = \order_1+\cdots+\order_\nbeq$.  Note that,
without loss of generality, we may truncate $\appbas$ so that $\deg(\appbas)
\le \max(\orders)$.

For example, suppose that the dimensions and the order are balanced: $\nbun =
\nbeq$ and $\orders = (\vsdim/\nbun,\ldots,\vsdim/\nbun)$.  Then,
$\cmat\in\matSpace[\nbun]$ is the coefficient of degree $\vsdim/\nbun$ of the
product $\appbas \sys$, where $\appbas$ and $\sys$ are $\nbun\times\nbun$
matrices over $\polRing$.  Thus $\cmat$ can be computed using $\vsdim/\nbun$
multiplications of $\nbun\times\nbun$ matrices over $\field$, at a total cost
$\bigO{\nbun^{\expmatmul-1} \vsdim}$.

Going back to the general case, the main obstacle to obtain similar efficiency
is that both the degrees in $\appbas$ and the order $\orders$ (hence the
degrees in $\sys$) may be unbalanced.  Still, we have $\cdeg{\sys} < \orders$
with sum $\sumVec{\orders}=\vsdim$ and, as stated in the introduction, we may
assume that either $\sumVec{\rdeg{\appbas}} \in \bigO{\vsdim}$ or
$\sumVec{\cdeg{\appbas}} \le \vsdim$ holds.  In this context, both $\sys$ and
$\appbas$ are represented by $\bigO{\nbun \vsdim}$ field elements.

We will generalize the method above for the balanced case to this general
situation with unbalanced degrees, achieving the same cost
$\bigO{\nbun^{\expmatmul-1} \vsdim}$.  As a result, computing the certificate
$\cmat$ has negligible cost compared to the fastest known approximant basis
algorithms.  Indeed, the latter are in $\softO{\nbun^{\expmatmul-1}\vsdim}$,
involving logarithmic factors in $\vsdim$ coming both from polynomial
arithmetic and from divide and conquer approaches.  We refer the reader to
\cite[Thm.\,5.3]{ZhoLab12} and \cite[Thm.\,1.4]{JeNeScVi16} for more details on
these logarithmic factors.

We first remark that $\cmat$ can be computed by naive linear algebra using
$\bigO{\nbun^2\vsdim}$ operations.  Indeed, writing $\rdeg \appbas = (r_1,
\ldots, r_\nbun)$, we have the following explicit formula for each entry in
$\cmat$:
\begin{equation*}
  \cmat_{i,j} = \sum_{k=1}^{\min(r_i,\order_j)}
  \appbas_{i,*,k} \, \sys_{*,j,\order_j-k} \;,
\end{equation*}
where $\appbas_{i,*,k}$ is the coefficient of degree $k$ of the row $i$ of
$\appbas$ and similar notation is used for $\sys$.  Then, since
$\min(r_i,\order_j)\leq\order_j$, the column $\matcol{\cmat}{j}$ is computed
via $\nbun\order_j$ scalar products of length $\nbun$, using
$\bigO{\nbun^2\order_j}$ operations. Summing this for $1\le j\le\nbeq$
yields $\bigO{\nbun^2\vsdim}$.

This approach considers each column of $\sys$ separately, allowing us to
truncate at precision $\order_j+1$ for the column $j$ and thus to rule out the
issue of the unbalancedness of the degrees in $\appbas$.  However, this also
prevents us from incorporating fast matrix multiplication.  In our efficient
method, we avoid considering columns or rows separately, while still managing
to handle the unbalancedness of the degrees in both $\appbas$ and $\sys$.  Our
approach bears similarities with algorithms for polynomial matrix
multiplication with unbalanced degrees (see for example
\cite[Sec.\,3.6]{ZhLaSt12}).

\subsection{Sparsity and degree structure}
\label{sec:compute_certificate:idea}

\emph{Below, we first detail our method assuming $\sumVec{\rdeg{\appbas}} \in
\bigO{\vsdim}$; until further notice, $\gamma \ge 1$ is a real number such that
$\sumVec{\rdeg{\appbas}} \le \gamma\vsdim$.}

To simplify the exposition, we start by replacing the tuple $\orders$ by the
uniform bound $\order = \max(\orders)$.  To achieve this, we consider the
matrix $\syss = \sys \xDiag{\order-\orders}$, where $\order-\orders$ stands for
$(\order-\order_1,\ldots,\order-\order_\nbeq)$: then, $\cmat$ is the
coefficient of degree $\order$ in $\appbas \syss$.

Since $\cdeg{\sys} < \orders$, we have $\deg(\syss)<\order$.  The fact that
$\sys$ has column degree less than $\orders$ translates into the fact that
$\syss$ has column valuation at least $\order-\orders$ (and degree less than
$\order$); like $\sys$, this matrix $\syss$ is represented by $\nbun \vsdim$
field elements. Recalling the assumption $\deg(\appbas) \le \order$, we can
write $\appbas = \sum_{k=0}^{\order} \appbas_k \var^k$ and $\syss =
\sum_{k=0}^{\order} \syss_k \var^k$, where $\appbas_k \in \matSpace[\nbun]$ and
$\syss_k \in \matSpace[\nbun][\nbeq]$ for all $k$ (note that $\syss_\order =
\matz$).  Then, our goal is to compute the matrix
\begin{equation}
  \label{eqn:cmat}
  \cmat = \sum_{k=1}^{\order} \appbas_k \syss_{\order-k}.
\end{equation}

The essential remark to design an efficient algorithm is that each matrix
$\appbas_k$ has only few nonzero rows when $k$ becomes large, and each matrix
$\syss_{\order-k}$ has only few nonzero columns when $k$ becomes large.  To
state this formally, we define two sets of indices, for the rows of degree at
least $k$ in $\appbas$ and for the orders at least $k$ in $\orders$:
\begin{align*}
  \largeRows{k} & = \{ i \in \{1,\ldots,\nbun\} \mid \rdeg{\matrow{\appbas}{i}} \ge k \}, \\
  \largeOrders{k} & = \{ j \in \{1,\ldots,\nbeq\} \mid \order_j \ge k \}.
\end{align*}
The latter corresponds to the set of indices of columns of $\sys$ which are
allowed to have degree $\ge k-1$ or, equivalently, to the set of indices of
columns of $\syss$ which are allowed to have valuation $\le\order-k$.

\begin{lemma}
  \label{lem:control_nrows_ncols}
  For a given $k \in \{1,\ldots,\order\}$: if $i \not\in \largeRows{k}$, then
  the row $i$ of $\appbas_k$ is zero; if $j \not\in \largeOrders{k}$, then the
  column $j$ of $\syss_{\order-k}$ is zero.  In particular, $\appbas_k$ has at
  most $\#\largeRows{k} \le \gamma\vsdim/ k$ nonzero rows and
  $\syss_{\order-k}$ has at most $\#\largeOrders{k} \le \vsdim/k$ nonzero
  columns.
\end{lemma}
\begin{proof}
  The row $i$ of $\appbas_k$ is the coefficient of degree $k$ of the row $i$ of
  $\appbas$.  If it is nonzero, we must have $i \in \largeRows{k}$.  Similarly,
  the column $j$ of $\syss_{\order-k}$ is the coefficient of degree $\order-k$
  of the column $j$ of $\syss=\sys\xDiag{\order-\orders}$.  If it is nonzero,
  we must have $\order-k \ge \order-\order_j$, hence $k \in \largeOrders{k}$.  

  The upper bounds on the cardinalities of $\largeRows{k}$ and
  $\largeOrders{k}$ follow by construction of these sets: we have $k \cdot
  \#\largeOrders{k} \le \sumVec{\orders} = \vsdim$, and also $k \cdot
  \#\largeRows{k} \le \sumVec{\rdeg{\appbas}}$ with $\sumVec{\rdeg{\appbas}}
  \le \gamma \vsdim$ by assumption.
\end{proof}

\subsection{Algorithm and cost bound}
\label{sec:compute_certificate:algo}

Following \cref{lem:control_nrows_ncols}, in the computation of $\cmat$ based
on \cref{eqn:cmat} we may restrict our view of $\appbas_k$ to its submatrix
with rows in $\largeRows{k}$, and our view of $\syss_k$ to its submatrix with
columns in $\largeOrders{k}$.  For example, if $k>\gamma\vsdim / \nbun$ and
$k>\vsdim/\nbeq$, the matrices in the product $\appbas_k \syss_k$ have
dimensions at most $\lfloor\gamma\vsdim/k\rfloor \times \nbun$ and $\nbun
\times\lfloor\vsdim/k\rfloor$.  These remarks on the structure and sparsity of
$\appbas_k$ and $\syss_k$ lead us to \cref{algo:certif_comp}.

\begin{algobox}
  \algoInfo
  {CertificateComp}
  {algo:certif_comp}

  \dataInfos{Input}{
    \item order $\orders \in \orderSpace$,
    \item matrix $\sys \in \sysSpace$ such that $\cdeg{\sys} < \orders$,
    \item matrix $\appbas \in \appbasSpace$ such that $\deg(\appbas) \le \max(\orders)$.
  }

  \dataInfo{Output}{
    the coefficient $\cmat\in \cmatSpace$ of degree $0$
    of $\appbas\sys \xDiag{-\orders}$.
  }

  \algoSteps{
    \item $(r_1,\dots,r_\nbun) \assign \rdeg{\appbas}$
    \item $\cmat \assign \matz \in \cmatSpace$
    \item \algoword{For} $k$ \algoword{from } $1$ \algoword{to } $\max(\orders)$: \\
      \phantom{For} $\largeRows{}  \assign \{ i \in \{1,\ldots,\nbun\} \mid r_i \ge k \}$ \\
      \phantom{For} $\largeOrders{} = \{c_1, \ldots , c_t\} \assign \{ j \in \{1,\ldots,\nbeq\} \mid \order_j \ge k \}$\\
      \phantom{For} $\mat{A} \in \matSpace[\#\largeRows{}][\nbun] \assign$ coefficient of degree $k$ of $\matrow{\appbas}{\largeRows{}}$ \\
      \phantom{For} $\mat{B} \in \matSpace[\nbun][t] \assign$ for all $1 \le j \le t$,
              $\matcol{\mat{B}}{j}$ is the coefficient of \\
              \phantom{For $\mat{B} \in \matSpace[\nbun][t] \assign$} 
              degree $\order_j-k$ of $\matcol{\sys}{c_j}$  \\
      \phantom{For} $\submat{\cmat}{\largeRows{}}{\largeOrders{}} \assign \submat{\cmat}{\largeRows{}}{\largeOrders{}} + \mat{A} \mat{B}$
    \item \algoword{Return} $\cmat$
  }
\end{algobox}

\begin{proposition}
  \label{prop:algo:certif_comp}
  \cref{algo:certif_comp} is correct.  Assuming that $\nbun\in\bigO{\vsdim}$
  and $\sumVec{\rdeg{\appbas}} \in \bigO{\vsdim}$, where
  $\vsdim=\sumVec{\orders}$, it uses $\bigO{\nbun^{\omega-1}\vsdim}$ operations
  in $\field$ if $\nbeq\le\nbun$ and
  $\bigO{\nbun^{\omega-1}\vsdim\log(\nbeq/\nbun)}$ operations in $\field$ if
  $\nbeq>\nbun$.
\end{proposition}
\begin{proof}
  For the correctness, note that for all $j$ the coefficient of degree
  $\order_j-k$ of $\matcol{\sys}{j}$ is the coefficient of degree $\order-k$ of
  $\matcol{\syss}{j}$.  Thus, using notation from
  \cref{sec:compute_certificate:idea}, the matrix $\mat{B}$ at the iteration
  $k$ of the loop is exactly the submatrix of $\syss_{\order-k}$ of its columns
  in $\largeOrders{k}$.  Therefore, the loop in \cref{algo:certif_comp} simply
  applies \cref{eqn:cmat}, discarding from $\appbas_k$ and $\sys_{\order-k}$
  rows and columns which are known to be zero.

  Now, we estimate the cost of updating $\cmat$ at each iteration of the loop.
  Precisely, the main task is to compute $\mat{A} \mat{B}$, where the matrices
  $\mat{A}$ and $\mat{B}$ have dimensions $\#\largeRows{} \times \nbun$ and
  $\nbun \times t$.  Then, adding this product to the submatrix
  $\submat{\cmat}{\largeRows{}}{\largeOrders{}}$ only costs $\#\largeRows{}
  \cdot t$ additions in $\field$.

  Consider $\gamma = \lceil\sumVec{\rdeg{\appbas}} / \vsdim\rceil \ge 1$
  (indeed, if $\sumVec{\rdeg{\appbas}} = 0$, then $\appbas$ is constant and
  $\cmat=\matz$).  By \cref{lem:control_nrows_ncols}, at the iteration $k$ we
  have $\#\largeRows{} \le \min( \nbun, \gamma\vsdim/ k)$ and $t =
  \#\largeOrders{} \le \min(\nbeq,\vsdim/k)$. We separate the cases $\nbeq \le
  \nbun$ and $\nbeq > \nbun$, and we use the bound
  $\lceil\gamma\vsdim/\nbun\rceil \in \bigO{\vsdim/\nbun}$, which comes from
  our assumptions $\nbun \in \bigO{\vsdim}$ and $\gamma \in \bigO{1}$.

  First, suppose $\nbeq\le \nbun$.  At the iterations $k < \lceil \gamma
  \vsdim/\nbun\rceil$ the matrices $\mat{A}$ and $\mat{B}$ both have dimensions
  at most $\nbun\times\nbun$, hence their product can be computed in
  $\bigO{\nbun^\expmatmul}$ operations.  These iterations have a total cost of
  $\bigO{\nbun^\expmatmul \lceil\gamma\vsdim/\nbun\rceil} \subseteq
  \bigO{\nbun^{\expmatmul-1} \vsdim}$.  At the iterations $k \ge
  \lceil\gamma\vsdim/\nbun\rceil$, $\mat{A}$ and $\mat{B}$ have dimensions at
  most $(\gamma\vsdim/k) \times \nbun$ and $\nbun \times (\vsdim/k)$, with
  $\vsdim/k \le \gamma\vsdim/k \le \nbun$; computing their product costs
  $\bigO{(\vsdim/k)^{\expmatmul-1} \nbun} \subseteq \bigO{\nbun
  \vsdim^{\expmatmul-1} k^{1-\expmatmul}}$.  Thus, the total cost for these
  iterations is in
  \begin{align*}
    & \bigOPar{
      \nbun \vsdim^{\expmatmul-1} \sum_{k=\lceil\gamma\vsdim/\nbun\rceil}^{\max(\orders)} k^{1-\expmatmul}
    } \\
    & \subseteq
    \bigOPar{
      \nbun \vsdim^{\expmatmul-1} (\lceil\gamma\vsdim/\nbun\rceil)^{2-\expmatmul} \textstyle\sum_{i=0}^{+\infty}  2^{i(2-\expmatmul)} 
    }
    \subseteq
    \bigO{\nbun^{\expmatmul-1} \vsdim}.
  \end{align*}
  For the first inclusion, we apply \cref{lem:bound_on_sum_powers} with $\mu =
  \lceil\gamma\vsdim/\nbun\rceil$, $\nu = \max(\orders)$, and $\theta = 1 -
  \expmatmul$. For the second, the sum is finite since $2^{2-\expmatmul} < 1$.
  Hence \cref{algo:certif_comp} costs $\bigO{\nbun^{\expmatmul-1} \vsdim}$ in
  the case $\nbeq\le\nbun$.

  Now, suppose $\nbeq>\nbun$.  At the iterations $k <
  \lceil\vsdim/\nbeq\rceil$, $\mat{A}$ and $\mat{B}$ have dimensions at most
  $\nbun\times\nbun$ and $\nbun\times\nbeq$, hence their product can be
  computed in $\bigO{\nbun^{\expmatmul-1}\nbeq}$.  The total cost is in
  $\bigO{\nbun^{\expmatmul-1}\vsdim}$ since there are $\lceil\vsdim/\nbeq\rceil
  - 1 < \vsdim/\nbeq$ iterations (with $\nbeq\le\vsdim$ by definition).  For
  the iterations $k \ge \lceil\gamma\vsdim/\nbun\rceil$, we repeat the analysis
  done above for the same values of $k$: these iterations cost
  $\bigO{\nbun^{\expmatmul-1}\vsdim}$ here as well.

  Finally, for the iterations $\lceil\vsdim/\nbeq\rceil \le k <
  \lceil\gamma\vsdim/\nbun\rceil$, $\mat{A}$ and $\mat{B}$ have dimensions at
  most $\nbun\times\nbun$ and $\nbun\times(\vsdim/k)$, with $\vsdim/k \le
  \nbeq$.  Thus the product $\mat{A} \mat{B}$ can be computed in
  $\bigO{\nbun^\expmatmul + \nbun^{\expmatmul-1}\vsdim/k}$ operations.  Summing
  the term $\nbun^\expmatmul$ over these $\bigO{\vsdim/\nbun}$ iterations
  yields the cost $\bigO{\nbun^{\expmatmul-1} \vsdim}$.  Summing the other term
  gives the cost $\bigO{\nbun^{\expmatmul-1} \vsdim \log(\nbeq/\nbun)}$
  since, by the last claim of \cref{lem:bound_on_sum_powers}, we have
  \[
    \sum_{k=\lceil\vsdim/\nbeq\rceil}^{\lceil\gamma\vsdim/\nbun\rceil-1} k^{-1}
    \;\le\;
    1+\left\lfloor\log_2\left(\frac{\lceil\gamma\vsdim/\nbun\rceil-1}{\lceil\vsdim/\nbeq\rceil}\right)\right\rfloor
    \;\le\;
    1 + \log_2(\gamma\nbeq/\nbun).
  \]
  Adding the costs of the three considered sets of iterations, we obtain the
  announced cost for \cref{algo:certif_comp} in the case $\nbeq>\nbun$ as well.
\end{proof}

\begin{lemma}
  \label{lem:bound_on_sum_powers}
  Given integers $0<\mu<\nu$ and a real number $\theta \le 0$,
  \[
    \sum_{k=\mu}^{\nu} k^\theta
    \;\;\le\;\;
    \mu^{\theta+1} \sum_{i=0}^{\ell-1} 2^{i(\theta+1)}
  \]
  holds, where $\ell = \lfloor\log_2(\nu/\mu)\rfloor+1$.  In particular,
  $\sum_{k=\mu}^{\nu} k^{-1} \le \ell$.
\end{lemma}
\begin{proof}
  Note that $\ell$ is chosen such that $2^{\ell}\mu -1 \ge \nu$.  Then, the
  upper bound is obtained by splitting the sum as follows:
  \[
    \sum_{k=\mu}^{\nu} k^\theta
    \le \sum_{i=0}^{\ell-1} \, \sum_{k=2^i\mu}^{2^{i+1}\mu-1} k^\theta
    \le \sum_{i=0}^{\ell-1} \, \sum_{k=2^i\mu}^{2^{i+1}\mu-1} (2^i\mu)^\theta 
    = \sum_{i=0}^{\ell-1} (2^i\mu)^{\theta+1},
  \]
  where the second inequality comes from the fact that $x \mapsto x^\theta$ is
  decreasing on the positive real numbers.
\end{proof}

Finally, we describe minor changes in \cref{algo:certif_comp} to deal with the
case of small average \emph{column} degree $\cdeg{\appbas}\in\bigO{\vsdim}$;
precisely, we replace the assumption $\sumVec{\rdeg{\appbas}} \le \gamma
\vsdim$ by $\sumVec{\cdeg{\appbas}} \le \gamma \vsdim$.  Then, instead of the
set $\largeRows{k}$ used above, we rather define
\[
  \largeCols{k} = \{j\in \{1,\dots,\nbun\} \mid \cdeg{\appbas_{*,j}} \geq k\}.
\]
Then we have the following lemma, analogous to \cref{lem:control_nrows_ncols}.
\begin{lemma}
  \label{lem:control_ncols_appbas}
  For $k \in \{1,\ldots,\nbun\}$ and $j \not\in \largeCols{k}$, the column $j$
  of $\appbas_k$ is zero.  In particular, $\appbas_k$ has at most
  $\#\largeCols{k} \le \gamma\vsdim/ k$ nonzero columns.
\end{lemma}

Thus, we can modify \cref{algo:certif_comp} to take into account the column
degree of $\appbas$ instead of its row degree.  This essentially amounts to
redefining the matrices $\mat{A}$ and $\mat{B}$ in the loop as follows:
\begin{itemize}
  \item  $\mat{A} \in \matSpace[\nbun][\#\largeCols{k}]$ is the
    coefficient  of degree  $k$ of  $\matcol{\appbas}{\largeCols{k}}$.
  \item $\mat{B} \in \matSpace[\#\largeCols{k}][t]$ is such that for all
    $i\in\largeCols{k}$ and $1 \le j \le t$, $\mat{B}_{i,j}$ is the coefficient
    of degree $\order_j-k$ of $\sys_{i,c_j}$  
\end{itemize}

These modifications have obviously no impact on the correctness.  Furthermore,
it is easily verified that the same cost bound holds since we obtain a similar
matrix multiplication cost at each iteration.

\section{Perspectives}
\label{sec:perspectives}

As noted in the introduction, our certificate is almost optimal since we can
verify it at a cost $\bigO{\nbun\vsdim+\nbun^{\expmatmul-1}(\nbun+\nbeq)}$
while the input size is $\nbun\vsdim$.  One should notice that the extra term
$\bigO{\nbun^{\expmatmul-1}(\nbun+\nbeq)}$ corresponds to certifying problems
of linear algebra over $\field$, namely the rank and the determinant.  These
could actually be dealt with in $\bigO{\nbun(\nbun+\nbeq)}$ operations using
interactive certificates built upon the results in
\cite{KalNehSau:2011,DumKal:2014,DumLucPer:2017}, thus yielding an optimal
certificate. Still, for practical applications, our simpler certification
should already be significantly faster than the approximant basis computation,
since the constants involved in the cost are small as we have observed in our
estimates above.  We plan to confirm this for the approximant bases
implementations in the LinBox library.

Finally, our verification protocol needs $(m+2)\log_2(\#S)$ random bits,
yielding a probability of failure less than $\frac{\vsdim+1}{\#\fsubset}$.  The
majority of these bits is required by \cref{algo:verif_truncmatprod} when
choosing $\nbun$ random elements for the vector $\row{u}$. As proposed in
\cite{KIMBREL1993}, it may be worthwhile to pick a single random value $\zeta$
and to use $\row{u}=[1 \;\, \zeta \;\, \cdots \;\, \zeta^{\nbun-1}]$. In the
case where $\max(\orders) < \vsdim/2$, this choice would not affect the
probability of failure while decreasing the number of random bits to
$3\log_2(\#\fsubset)$.  In particular, at the price of the same number of bits
as we currently use in our algorithm, we could run our verification
$(\nbun+2)/3$ times and decrease the probability of failure to
$(\frac{\vsdim+1}{\#\fsubset})^{\frac{\nbun+2}{3}}$.


\end{document}